\documentclass[twocolumn,showpacs,aps,pra,amsmath,amssymb,superscriptaddress]{revtex4-2}
\usepackage{bm,color,bbm}
\usepackage{hyperref, mathtools,graphicx,natbib,amssymb,amsthm}
\newtheorem{theorem}{Theorem}
\newtheorem{proposition}{Proposition}
\newtheorem{corollary}{Corollary}

\newcommand{\beq}{\begin{equation}}
\newcommand{\eeq}{\end{equation}}
\newcommand{\bqa}{\begin{eqnarray}}
\newcommand{\eqa}{\end{eqnarray}}
\newcommand{\nn}{\nonumber}

\newcommand{\smallfrac}[2]{\mbox{$\frac{#1}{#2}$}}

\newcommand{\half}{\smallfrac{1}{2}}

\newcommand{\sq}[1]{\left[ {#1} \right]}

\newcommand{\tr}[1]{{\rm Tr}\sq{ {#1} }}

\newcommand{\id}{\mathbbm{1}}

\newcommand{\h}{\mathcal{H}}
\newcommand{\p}{\mathcal{P}}
\newcommand{\ens}{\mathcal{E}}

\newcommand{\blk}{\color{black}}

\definecolor{maroon}{rgb}{0.7,0,0}

\definecolor{ngreen}{rgb}{0.3,0.7,0.3}

\definecolor{golden}{rgb}{0.8,0.6,0.1}

\begin{document}
\title{
Asymmetry and tighter uncertainty relations for R\'enyi entropies via quantum--classical decompositions of resource measures}
\author{Michael J. W. Hall}
\affiliation{Department of Theoretical Physics, Research School of Physics, Australian National University, Canberra, ACT 0200, Australia}

\begin{abstract}
It is known that the variance and entropy of quantum observables decompose into intrinsically quantum and classical contributions. Here a general method of constructing quantum--classical decompositions of resources such as uncertainty is discussed, with the quantum contribution specified by a measure of the noncommutativity of a given set of operators relative to the quantum state, and the classical contribution generated by the mixedness of the state. Suitable measures of noncommutativity or `quantumness' include quantum  Fisher information, and the asymmetry of a given set, group or algebra of operators, and are  generalised to nonprojective observables and quantum channels. Strong entropic uncertainty relations and lower bounds for R\'enyi entropies are  obtained, valid for arbitrary discrete observables, that take the mixedness of the state into account via a classical contribution to the lower bound.  These relations can also be interpreted without reference to quantum-classical decompositions, as tradeoff relations that bound the asymmetry of one observable in terms of the entropy of another.
\end{abstract}

\maketitle

\section{Introduction}

A quantum observable such as position or spin has two generic sources of uncertainty. The first arises when the observable does not commute with the state of the quantum system: this implies that the state is not an eigenstate of the observable, leading to a spread of measurement outcomes. The second source of uncertainty arises when the state is mixed: a loss of classical information due to mixing reduces the predictability of the measurement outcome, leading to an increased spread.  It follows that the two sources are,  respectively, intrinsically quantum and classical in nature.  The first can be useful as a resource, e.g., in allowing a complementary observable to have a small spread, whereas the second is usually not, as mixing contributes noise to all observables.

Luo suggested quantifying such quantum and classical contributions to a given measure of uncertainty, $M(X|\rho)$, for a Hermitian operator $X$ and state described by density operator $\rho$,  via a quantum-classical decomposition of the form~\cite{Luo2005,LuoPRA2005} 
\beq \label{decomp}
M(X|\rho) = Q(X|\rho) + C(X|\rho) ,
\eeq
and gave the particular example 
\beq \label{skewex}
{\rm Var}_\rho(X) = -\frac12\tr{ [X,\sqrt{\rho}]^2} + \left(\tr{X\sqrt{\rho}X\sqrt{\rho}} - \langle X\rangle_\rho^2\right)
\eeq
for the variance of a Hermitian operator. In this example the `quantum' contribution is the skew information of $X$ with respect to $\rho$~\cite{WYD}, and is clearly nonzero only if the observable does not commute with the state, whereas the `classical' contribution is nonzero only if the state is mixed, as expected. A second example is a quantum--classical decomposition of the entropy of a Hermitian operator, given by Korzekwa~{\it et al.}~\cite{kor}.

A simple approach to unifying and generalising such quantum--classical decompositions is discussed in Sec.~\ref{sec:decomp}, based on the idea that the intrinsically quantum contribution to a given resource measure  is maximised when the observer has access to a pure state of the system, and is degraded by classical mixing for a local observer who only has access to a component of such a system. For example, a maximally-entangled pure state of two qubits is a useful quantum resource for dense coding and for estimation of a local rotation, but this usefulness vanishes for an observer who only has access to one of the two qubits (described by a maximally-mixed state).  

The general approach is described in Sec.~\ref{sec:decompgen}. It starts with some given measure $Q(X,Y,\dots|\rho)$ of the `quantumness' of a resource, for operators $X,Y,\dots$ and state $\rho$. The corresponding  `maximum potential' $M(X,Y,\dots|\rho)$ of the resource and its `classicality' $C(X,Y,\dots|\rho)$ are then constructed so as to satisfy a decomposition analogous to Eq.~(\ref{decomp}). Several examples are given in Sec.~\ref{sec:decompex}, including the decomposition of variance and covariance matrices with respect to quantum Fisher information; decompositions of Shannon entropy with respect to both asymmetry and the conditional entropy of self-dual communication channels; and, of particular relevance to this paper, the decomposition of R\'enyi entropy with respect to R\'enyi asymmetry, where the latter has applications to quantum coherence, time-energy uncertainty relations, quantum information, open quantum systems and quantum metrology ~\cite{lara,ColesPRL,Hall2022,Chitambar,Gao2020}. It is also shown that measures of quantumness can equivalently be defined via sets, groups or algebras of operators, and generalised to  arbitrary discrete observables and to quantum channels. An alternative approach, based on convex and concave roofs, is discussed in Appendix~\ref{appa}, with corresponding examples.

The main technical results are given in Sec.~\ref{sec:uncert}, including a strong uncertainty relation for R\'enyi entropies of discrete  observables $X$ and $Y$ represented by positive operator valued measures (POVMs) $\{X_x\}$ and $\{Y_y\}$:
\begin{align}  \label{uncertrenyi}
H_\alpha(X|\rho)+H_\beta(Y|\rho) &\geq -\log \mu_{XY}\nn\\&~\,~~+\max\{C_\alpha(Y|\rho),C_\beta(X|\rho)\} .
\end{align}
Here the R\'enyi parameters $\alpha,\beta\in[\half,\infty)$ are related by $\alpha^{-1}+\beta^{-1}=2$; $\mu_{XY}$ denotes the maximum eigenvalue of $X_x^{1/2}Y_yX_x^{1/2}$ over $x$ and $y$; and the last term takes the mixedness of $\rho$ into account via the classical components of the R\'enyi entropies. The logarithm base is left arbitrary throughout, corresponding to a choice of units (e.g., to bits for base 2 and to nats for base $e$). 

 Equation~(\ref{uncertrenyi}) is clearly stronger than the standard uncertainty relation
\beq \label{renknown}
H_\alpha(X|\rho)+H_\beta(Y|\rho) \geq -\log \mu_{XY} ,
\eeq
for  R\'enyi entropies~\cite{maas,Para,Rastegin2010}, which has no classical mixing term. Further, for   
rank-1 POVMs, with $X_x=|x\rangle\langle x|$ and $Y_y=|y\rangle\langle y|$,  Eq.~(\ref{uncertrenyi}) generalises the known uncertainty relation
\beq \label{nondegen}
H(X|\rho) + H(Y|\rho) \geq -\log \max_{x,y} |\langle x|y\rangle|^2+ H(\rho)
\eeq
for Shannon entropies~\cite{kor,Berta,Coles2011} (which correspond to $\alpha=\beta=1$), that takes the mixedness of the state into account via its von Neumann entropy, $H(\rho):=-\tr{\rho\log\rho}$.   It is also shown in Sec.~\ref{sec:uncert} that Eq.~(\ref{uncertrenyi}) may be written as a direct tradeoff between the asymmetry of one observable and the entropy of another, without any reference to quantum-classical decompositions, leading to a simple lower bound for R\'enyi asymmetry, as well as to entropic bounds such as
\beq \label{rhoxbound}
H_\alpha(X|\rho) \geq -\log\mu_{\rho X},\qquad \alpha\geq\half ,
\eeq
for the R\'enyi entropy of an arbitrary discrete observable $X$.   Some readers may choose to skip directly to Sec.~\ref{sec:uncert}, and refer back to Sec.~\ref{sec:decomp} for definitions.

Results and possible future work are discussed in Sec.~\ref{sec:con}.

\section{Building quantum--classical decompositions }
\label{sec:decomp}

\subsection{A general approach}
\label{sec:decompgen}


As noted in the Introduction, access to a maximally-entangled pure state of two qubits is a useful resource for the  estimation of a local rotation, but this usefulness vanishes for an observer who only has access to one of the two qubits, described by a maximally-mixed state. This motivates the more general 
idea that the quantum contribution to the degree to which a given resource can be exploited is maximised for a notional observer who has access to a pure state of the system, with any classical component arising from the presence of mixing when there is only partial access. 

Accordingly, if the `quantumness' of some resource is quantified by a measure $Q(X,Y,\dots|\rho)$, for some set of operators $X,Y,\dots$ and state $\rho$, the  maximum potential of the resource, $M(X,Y,\dots|\rho)$, is defined by
\beq \label{ult}
M(X,Y,\dots|\rho) := Q(X\otimes\id_a,Y\otimes \id_a,\dots|\rho_\psi),
\eeq
where $\rho_\psi\equiv|\psi\rangle\langle\psi|$ is a purification of state $\rho$ on the tensor product  $\h\otimes \h_a$ of the system Hilbert space $\cal H$ with an ancillary Hilbert space $\h_a$~\cite{nielsen}. This immediately yields a corresponding quantum-classical decomposition of the form
\begin{align} \label{decompgen}
	M(X,Y,\dots|\rho) = Q(X,Y,\dots|\rho) + C(X,Y,\dots|\rho) ,
\end{align}
generalising Eq.~(\ref{decomp}), with $C(X,Y,\dots|\rho) := Q(X\otimes\id_a,Y\otimes \id_a,\dots|\rho_\psi) - Q(X,Y,\dots|\rho)$ representing the classical component of the resource relative to an observer who only has access to the mixed state $\rho$ of the system (rather than to the purified state $\rho_\psi$).

Several examples and generalisations are given  below, but first three natural assumptions for  $Q(X,Y,\dots|\rho)$ are identified (related requirements have been previously discussed by Luo~\cite{Luo2005,LuoPRA2005} and Korzekwa {\it et al.}~\cite{kor}).
\begin{itemize}
		\item[(i)] It will be assumed that there is no potential `quantum' resource to exploit if the relevant operators commute with the state of the system, i.e.,  
		\beq \label{q0}
	\qquad	Q(X,Y,\dots|\rho) =0 ~~{\rm if}~~ [X,\rho]=[Y,\rho]=\dots=0 .
		\eeq
	Hence, the quantumness acts as a measure of noncommutativity. 
	\item[(ii)] 
	The decomposition must be well-defined, with the values of the quantum and classical components being independent of the choice of purification. This is guaranteed by the assumption that $Q(X,Y,\dots|\rho)$ is invariant under  unitary transformations, i.e., 
	\beq \label{assump1}
\qquad	Q(X,Y,\dots|\rho) = Q(UXU^\dagger,UYU^\dagger,\dots|U\rho U^\dagger) 
	\eeq
	for all unitary transformations $U$. In particular, noting that any two purifications $\rho_\psi$ and $\rho_{\psi'}$ can be related by a local unitary transformation $U=\id\otimes U_a$ acting on the tensor product  $\h\otimes \h_a$~\cite{nielsen}, where this leaves operators on $\h$ invariant, it follows that Eqs.~(\ref{ult}) and Eq.~(\ref{decompgen}) are independent of the choice of purification as required. 
	\item[(iii)] 
	The classical component must be nonnegative, and vanish for pure states (i.e., when there is no classical mixing). This is equivalent, via  Eqs.~(\ref{ult}) and~(\ref{decompgen}),  to the assumption that the quantumness increases under purification (e.g., when an observer gains access to the full quantum system), i.e., that
	\beq \label{assump2}
	\qquad Q(X,Y,\dots|\rho) \leq Q(X\otimes\id_a,Y\otimes \id_a,\dots|\rho_\psi)
	\eeq
	for any purification $\rho_\psi$ of $\rho$. In particular, this assumption ensures that $C(X,Y,\dots|\rho)\geq0$, with equality for pure states, as desired.
\end{itemize}

The main advantage of the above construction method is that the maximum potential value of a quantum resource accessible to an observer, $M$, and the classical contribution due to mixing, $C$, are fully determined by the choice of the quantumness (or noncommutativity) measure $Q$. Other advantages are the applicability to arbitrary sets of operators and the straightforward extension to nonprojective observables and quantum channels, as illustrated in the examples below. An alternative  method of constructing quantum-classical decompositions, based on convex and concave roofs of pure-state resource measures, is discussed in Appendix~\ref{appa}.

\subsection{Examples}
\label{sec:decompex}

\blk 
It is convenient in what follows to call a given discrete observable $X$ a {\it projective}  observable if the corresponding POVM $\{X_x\}$ consists of orthogonal projections, i.e., $X_xX_{x'}=\delta_{xx'}X_x$, and a  {\it nonprojective} observable otherwise.  Note that a projective observable $X$ is equivalently represented by the  Hermitian operator $X\equiv\sum_x xX_x$. The distinction between a projective observable and its corresponding Hermitian operator will always be clear by context.
\blk

\subsubsection{Skew information, variance, and the \\Fisher information matrix}
\label{sec:skew}

Luo's decomposition of the variance of a projective observable into quantum and classical contributions~\cite{Luo2005,LuoPRA2005}, 
\beq \label{skewdecomp}
{\rm Var}_\rho(X)= M_{\rm skew}(X|\rho)= Q_{\rm skew}(X|\rho)+C_{\rm skew}(X|\rho)
	\eeq
 as per Eq.~(\ref{skewex}), corresponds to choosing the quantumness measure to be the skew information of the observable with respect to the state~\cite{WYD}, i.e.,
\beq \label{qskew}
Q_{\rm skew}(X|\rho):=\frac12\tr{ (i[X,\sqrt{\rho}])^2}  .
\eeq
The skew information may be physically interpreted as, for example, a measure of the information content of an ensemble that conserves $X$~\cite{WYD};  a variant of quantum Fisher information in quantum metrology~\cite{Luo2005,LuoPRL,Petz};  a measure of asymmetry~\cite{Marvian2014}; or (for nondegenerate observables) as a measure of coherence~\cite{Girolami2014, Luo2017}.

Note that assumptions~~(\ref{q0}) and~(\ref{assump1}) are trivially satisfied by the skew information, while assumption~(\ref{assump2}) follows because $Q_{\rm skew}(X\otimes\id_a|\rho_\psi)={\rm Var}_\rho(X)$ for any purification of $\rho$, implying via identity~(\ref{skewex}) and Eq.~(\ref{skewdecomp}) that
\beq \label{cskew}
C_{\rm skew}(X|\rho) = \tr{\left(\rho^{1/4}(X-\langle X\rangle_\rho )\rho^{1/4}\right)^2}\geq 0 .
\eeq
Hence, properties~(i)--(iii) of Sec.~\ref{sec:decompgen} all hold as required. 

The above decomposition can be generalised to the case of more than one projective observable, and to other measures of quantum Fisher information, using recent results by Kudo and Tajima~\cite{Kudo2022}. In particular, let $Q^f_F(X_1,X_2,\dots|\rho)$ denote the quantum Fisher information matrix corresponding to a given monotone metric function $f$, for some set of projective observables $X_1, X_2,\dots$ and density operator $\rho$, where each choice of $f$ corresponds to a measure of the sensitivity of the state to unitary transformations generated by these observables~\cite{Petz, Kudo2022}. Assumptions~(\ref{q0})--(\ref{assump2}) of Sec.~\ref{sec:decompgen} hold for $Q^f_F$ as a direct consequence of Theorem~9 of~\cite{Kudo2022}, and the construction method yields the quantum-classical decomposition 
\beq \label{covdecomp}
M(X_1,X_2,\dots|\rho) ={\rm Cov}_\rho(X_1,X_2\dots|\rho)= Q^f_F+C^f_F
\eeq
of the symmetrised covariance matrix, with coefficients defined by
\begin{align}
[{\rm Cov}_\rho(X_1,X_2,\dots|\rho)]_{jk}&:= \half\langle X_jX_k+X_kX_j\rangle_\rho\nn\\
&\qquad~~~-\langle X_j\rangle_\rho\langle X_k\rangle_\rho .
	\end{align}
Thus, the covariance matrix has an infinite family of quantum-classical decompositions, with Eqs.~(\ref{skewdecomp})--(\ref{cskew}) providing just one particular example. It is shown in Appendix~\ref{appa} that, for the case of a single projective observable, the alternative construction method given there  picks out the particular decomposition having the minimal Fisher information (corresponding to the symmetric logarithmic derivative~\cite{Petz}).

\subsubsection{Distance, Shannon entropy and standard asymmetry}
\label{sec:selfdual}

Shannon entropy is a well known resource measure in various information contexts, and Korzekwa {\it et al.} have given a corresponding quantum-classical decomposition for the case of projective observables~\cite{kor}, as briefly described below. An alternative decomposition of Shannon entropy is noted in Appendix~\ref{appa:selfdual}, based on information properties of self-dual quantum communication channels, \blk i.e., communication channels which are invariant under a duality mapping between the signal ensemble and the receiver measurement~\cite{Hugh93, Hall1997}. \blk

The decomposition in ~\cite{kor} corresponds to choosing the quantumness to be a particular measure of the `distance' from the state of the system to a set of postmeasurement states. More generally, let $d(\rho,\sigma)\geq 0$ be a positive function that vanishes for $\rho=\sigma$, and define the corresponding measure of quantumness for a projective observable $X$ with associated projection valued measure $\{X_x\}$ via
\beq \label{qd}
Q_d(X|\rho):=\inf_\sigma d(\rho,\sigma_X) = \inf_{\sigma:[\sigma,X]=0} d(\rho,\sigma),
\eeq
where $\sigma$ is implicitly restricted to the set of density operators of the system and 
\beq \label{sigmax}
\sigma_X:= \sum X_x \sigma X_x 
\eeq
is the postmeasurement state for a projective measurement of $X$ on state $\sigma$. Such distance-based measures of quantumness or noncommutativity are commonly used to quantify asymmetry and coherence resources~\cite{Marvian2014,coh}. 

Note that $Q_d$ satisfies assumption~(\ref{q0}) of Sec.~\ref{sec:decompgen} by construction,  while assumptions~(\ref{assump1}) and~(\ref{assump2}) are satisfied if $d(\rho,\sigma)$ is invariant under unitary transformations and is
nonincreasing under the partial trace operation---and in
particular if $d(\rho,\sigma)$ is nonincreasing under general completely positive trace preserving (CPTP) maps~\cite{nielsen}. These properties hold in the case of the relative entropy distance function~\cite{nielsen}, defined by
\beq \label{d1}
d(\rho,\sigma)= D_1(\rho\|\sigma):=\tr{\rho(\log\rho-\log\sigma)} ,
\eeq
for which $Q_d$ evaluates to the standard measure of asymmetry~\cite{vacc1,vacc2},
\beq \label{q1}
	Q_1(X|\rho) = D_1(\rho\|\rho_X) = H(\rho_X) - H(\rho) ,
\eeq
using the identity $D_1(\rho\|\sigma_X)=H(\rho_X)-H(\rho)+D_1(\rho_X\|\sigma_X)$. 
Here the subscript `1' provides the basis for a generalisation to R\'enyi asymmetry in the next example. The standard asymmetry $Q_1(X|\rho)$ is equal to the increase in system entropy due to a nonselective measurement of $X$, and is a useful quantum resource measure in various contexts~\cite{Marvian2014,vacc1, vacc2,Gour,coh, kor,HallPRX,aberg} (these contexts typically have different sets of `free' states and operations, but the resource measure itself has the same form).

Korzekwa~{\it et al.} simply postulate that the standard asymmetry, $Q_1$, is the intrinsically quantum contribution to the Shannon entropy of the observable,  $H(X|\rho):=-\sum_x p(x|\rho)\log p(x|\rho)$~\cite{kor}. In contrast, an advantage of the general approach in Sec.~\ref{sec:decompgen} is that the link to Shannon entropy is derived rather than postulated.  In particular,  the maximum potential of the standard asymmetry, available to an observer with access to a purification $|\psi\rangle\langle\psi|$ of $\rho$, follows from definition~(\ref{ult}) and Eq.~(\ref{q1}) as
\begin{align}
M_1(X|\rho) &= H\big(\sum_x X_x\otimes\id_a|\psi\rangle\langle\psi|\ X_x\otimes\id_a\big) - H(|\psi\rangle\langle\psi|) \nn\\
&=H\big(\sum_x\langle\psi|X_x\otimes \id_a|\psi\rangle |\psi_x\rangle\langle\psi_x|\big) \nn\\
&=H(X|\rho),
\end{align} 
where $\{|\psi_x\rangle\}$ is the set of orthonormal states with $|\psi_x\rangle$ proportional to $(X_x\otimes\id_a)|\psi\rangle$, and the last line follows via $\langle\psi|X_x\otimes \id_a|\psi\rangle=\tr{|\psi\rangle\langle\psi| X_x\otimes\id_a}=\tr{\rho X_x}=p(x|\rho)$. Hence, the approach in Sec.~\ref{sec:decompgen} constructively generates the 
quantum-classical decomposition 
\begin{align} \label{u1}
H(X|\rho) = Q_1(X|\rho)+ C_1(X|\rho)
\end{align}
given in~\cite{kor}, with the classical contribution evaluating to
\begin{align} \label{c1}
	C_1(X|\rho) &= H(\rho) - \sum_x p(x|\rho) H(X_x\rho X_x/p(x|\rho)).
	\end{align}
Note for any rank-1 projective observable $X$, with $X_x\equiv|x\rangle\langle x|$, that the classical contribution reduces to the von Neumann entropy $H(\rho)$~\cite{kor}. More generally, if $H(X|\rho)$  is identified with the entropy of the classical measurement record, then the classical contribution equals the average decrease of the system entropy due to making such a record.

Korzekwa {\it et al.} use decomposition~(\ref{u1}) to obtain a simple proof of entropic uncertainty relation~(\ref{nondegen}) for rank-1 projective observables~\cite{kor}. The decomposition, uncertainty relation and method of proof are generalised below to R\'enyi entropies and  to nonprojective observables. A rather different decomposition of Shannon entropy, based on a quantumness measure for self-dual communication channels, is discussed in Appendix~\ref{appa:selfdual}.

\subsubsection{R\'enyi asymmetry and R\'enyi entropy}
\label{sec:asymm}

A particular choice of interest for the distance function $d(\rho,\sigma)$ in Eq.~(\ref{qd}) is the quantum generalisation of the classical R\'enyi relative entropy~\cite{Renyi} to the sandwiched R\'enyi divergence~\cite{Muller2013, Wilde2014},
\beq \label{dalpha}
d(\rho,\sigma) = D_\alpha(\rho\|\sigma) := \frac{1}{\alpha-1} \log \tr{\big(\sigma^{\frac{1-\alpha}{2\alpha}}\rho \sigma^{\frac{1-\alpha}{2\alpha}}\big)^\alpha}, 
\eeq
where the index $\alpha$ ranges over $[0,\infty)$. This generalises the relative entropy function in Eq.~(\ref{d1}) (which corresponds to the limit $\alpha\rightarrow1$), and introductory expositions of its basic properties may be found, for example, in~\cite{Muller2013,Beigi2013}. The most important of these properties for the purposes of this paper is that it is nonincreasing under CPTP maps when $\alpha\geq\half$, i.e, the data-processing inequality
\beq \label{data}
D_\alpha(\phi(\rho)\|\phi(\sigma))\leq D_\alpha(\rho\|\sigma),\qquad \alpha\geq\half,
\eeq
holds for any CPTP map $\phi$~\cite{Beigi2013,Lieb}.

For projective observables  the above choice of $d(x,y)$ generates the corresponding quantumness measure~\cite{lara}
\beq \label{renq}
Q_\alpha(X|\rho) := \inf_{\sigma:[\sigma,X]=0} D_\alpha(\rho\|\sigma) = \inf_{\sigma} D_\alpha(\rho\|\sigma_X) 
\eeq
via~Eq.~(\ref{qd}). This measure generalises the standard asymmetry in Eq.~(\ref{q1}), and is known as the R\'enyi asymmetry of $X$ with respect to $\rho$. It has found applications in the areas of quantum coherence, quantum information,  time-energy uncertainty relations, quantum metrology and open quantum systems~\cite{lara,ColesPRL,Hall2022,Chitambar,Gao2020}.  

Assumptions~(\ref{q0})--(\ref{assump2}) hold for the R\'enyi asymmetry $Q_\alpha(X|\rho)$ when $\alpha\geq\half$ as a consequence of data-processing inequality~(\ref{data}), as may easily be checked. Further, the maximum potential of the R\'enyi asymmetry is directly related to the R\'enyi entropy~\cite{Renyi}
\beq \label{rendef}
H_\beta(X|\rho):= \frac{1}{1-\beta}\log \sum_xp(x|\rho)^\beta
\eeq
of observable $X$ for state $\rho$, via 
\beq \label{urenyi}
M_\alpha(X|\rho)=Q_\alpha(X\otimes\id_a|\rho_\psi)=H_\beta(X|\rho),~~\frac{1}{\alpha}+\frac{1}{\beta}=2 .
\eeq
Here the first equality follows via definition~(\ref{ult}) and the second equality by  showing,  via direct minimisation in Eq.~(\ref{renq}), that  $Q_\alpha(X|\,|\psi\rangle\langle\psi|)=H_\beta(|\psi\rangle\langle\psi|_X)$ for any pure state $|\psi\rangle$, from which $M_\alpha(X|\rho)=H_\beta((\rho_\psi)_{X\otimes\id_a})=H_\beta(X|\rho)$ follows by direct calculation (see also the proof of Theorem~3 in~\cite{Hall2022}). Hence,  the general
approach in Sec.~\ref{sec:decompgen}) yields the  quantum-classical decomposition
\beq \label{renyidecomp}
 H_\beta(X|\rho) = Q_\alpha(X|\rho) + C_\alpha(X|\rho), ~~\frac{1}{\alpha}+\frac{1}{\beta}=2,
\eeq
of R\'enyi entropy, generalising the case of Shannon entropy in Eq.~(\ref{u1}) (which corresponds to the limit $\alpha\rightarrow1$). It follows that the R\'enyi asymmetry $Q_\alpha(X|\rho)$ is upper bounded by the R\'enyi entropy $H_\beta(X|\rho)$, with equality for pure states~\cite{Hall2022}. In contrast, the classical component $C_\alpha(X|\rho)$ vanishes for pure states by construction, as befits a measure of mixedness, and reaches the upper bound of $H_\beta(X|\rho)$ when $X$ is `classical' with respect to $\rho$, i.e., when $[X,\rho]=0$. Decomposition~(\ref{renyidecomp}) will be used in Sec.~\ref{sec:uncert} to obtain strong uncertainty relation~(\ref{uncertrenyi}) for R\'enyi entropies.

\subsubsection{Generalising to nonprojective observables}
\label{sec:nonproj}

The previous examples are restricted to projective observables. However, it is not difficult to generalise to arbitrary observables, as required for the general results in Sec.~\ref{sec:uncert}.

The idea is to exploit Naimark's extension theorem, that an observable $X$ with POVM $\{ X_x\}$ on Hilbert space $\cal H$, can always be extended to a projective observable $\tilde X$ with projection valued measure $\{\tilde X_x$\} on a larger Hilbert space $\tilde{\cal H}$, with
\beq \label{naimark}
X_x = P \tilde X_x P,
\eeq
where $P$ denotes the projection from $\tilde{\cal H}$ onto $\cal H$~\cite{nai,hol,peres,footnai, foottrine}.  
It follows that any given measure of quantumness $Q(X|\rho)$ for projective observables can be extended to nonprojective observables via (noting $\tilde \rho\equiv\rho$ on $\tilde{\cal H}$ since $P\tilde\rho P=\rho$)
\beq \label{extension}
Q(X|\rho) := Q(\tilde X|\tilde\rho)\equiv Q(\tilde X|\rho) .
\eeq
While this definition will typically depend on the choice of extension mapping $X\rightarrow \tilde X$, the results obtained in this paper are valid for any choice. Hence, the mapping can and will be left unspecified in what follows. 

As a simple example, the extension of the skew information in Eq.~(\ref{qskew}) to general observables follows, using $P\rho P=\rho$ and $\tilde X_x\tilde X_{x'}=\delta_{xx'}\tilde X_x$, as
\beq
Q_{\rm skew}(X|\rho) = \frac12\tr{ (i[\,\overline{X},\sqrt{\rho}])^2} +
\tr{\rho (\overline{X^2} -\overline{X}^2)} ,
\eeq
 in terms of the moment operators $\overline{X^m}:=\sum_x x^m X_x$. This is independent of the mapping $X_x\rightarrow \tilde X_x$ and reduces to Eq.~(\ref{qskew}) for projective observables. Note that the second term is typically nonzero for nonprojective observables, corresponding to an additional information resource if the extended observable is physically accessible. Remarkably, the maximum potential of the resource is  given by the variance for both projective and nonprojective observables, with definition~(\ref{ult}) leading to 
 \beq
 M_{\rm skew}(X|\rho)={\rm Var}_\rho(X)=\langle \overline{X^2}\rangle_\rho -\langle \overline X\rangle_\rho^2  .
 \eeq
Thus, the general decomposition of variance has the same form as Eq.~(\ref{skewdecomp}), with $X$ replaced by $\bar X$ in Eq.~(\ref{cskew}).

As a second example, the R\'enyi asymmetry in Eq.~(\ref{renq}) generalises to
\beq \label{rengen}
Q_\alpha(X|\rho) = \inf_{\sigma:[\sigma,\tilde X]=0} D_\alpha(\rho\|\sigma) =\inf_\sigma D_\alpha(\rho\|\sigma_{\tilde X}) 
\eeq
where, as always, $\sigma$ is implicitly restricted to range over the density operators of the system (hence, $P\sigma P=\sigma$). Thus, the asymmetry quantifies the distance from the state of the system to the set of postmeasurement states following a measurement of $\tilde X$. Such postmeasurement states do not typically lie in the Hilbert space $\cal H$ of the system, and hence the asymmetry will typically be nonvanishing for nonprojective observables (see also Corollary~\ref{corinsert}), corresponding to the sensitivity of $\rho$ to unitary displacements generated by $\tilde X$ on the extended Hilbert space. If such displacements are not physically accessible then the generalised asymmetry might more appropriately be referred to as hidden asymmetry (or perhaps super-asymmetry --- bazinga!).  

Note from Eq.~(\ref{naimark}) and the property $P\rho P=\rho$ that ${\rm Tr}[\rho \tilde X_x]=\tr{\rho X_x}$, implying that $H_\beta(\tilde X|\rho)=H_\beta(X|\rho)$. It follows immediately that Eq.~(\ref{rengen}) generates a quantum-classical decomposition of R\'enyi entropy having precisely the same form as Eq.~(\ref{renyidecomp}) for projective observables. Further, it is straightforward to show, using definition~(\ref{rengen}) and $P|\psi\rangle=|\psi\rangle$, that for a pure state $\rho=|\psi\rangle\langle\psi|$ the generalised asymmetry is given by 
\beq \label{asymmgen}
Q_\alpha(X|\,|\psi\rangle\langle\psi|)=H_\beta(X|\,|\psi\rangle\langle\psi|),~~~\frac{1}{\alpha}+\frac{1}{\beta}=2,
\eeq
similarly to the case of projective observables. Hence,  Eq.~(\ref{renyidecomp}) for the general case implies that the classical contribution to the R\'enyi entropy vanishes for pure states for both projective and nonprojective observables, i.e.
\beq \label{classgen}
C_\alpha(X|\,|\psi\rangle\langle\psi|) = 0,\qquad \alpha\geq \half.
\eeq

Finally, it is of interest to note that the general definition of asymmetry in Eq.~(\ref{rengen}) reduces to Eq.~(\ref{renq}) for the case of projective observables, irrespective of the choice of extension mapping $X_x\rightarrow \tilde X_x$. In particular, if $PEP$ is a projection for two projections $P$ and $E$, then writing $P^\perp =\id-P$ one has $(PEP^\perp )(PEP^\perp )^\dagger =PEP-(PEP)^2=0$, and so $PEP^\perp =0=P^\perp EP$, yielding
$E=(P+P^\perp )E(P+P^\perp )=PEP+P^\perp EP^\perp $,  from which $[P,E]=0$ follows. Hence, if $X$ is a projective observable then $[P,\tilde X_x]=0$ via Eq.~(\ref{naimark}). But $\rho$ and $\sigma$ in Eq.~(\ref{rengen}) similarly commute with $P$ (since $P\rho P= \rho$ and $P\sigma P=\sigma$), and evaluation of the right hand side then leads directly to Eq.~(\ref{renq}) as claimed.

\subsubsection{Asymmetry of sets vs groups vs algebras vs channels}
\label{sec:chann}

It is straighforward to generalise the class of distance-based quantumness measures in Eq.~(\ref{qd}) from the case of a single projective observable $X$ to any set $S$ of bounded operators, via
\beq \label{qds}
Q_d(S|\rho) := \inf_{\sigma: [\sigma,X]=0\,\forall X\in S} d(\rho,\sigma) = \inf_{\sigma\in S'} d(\rho,\sigma),
\eeq
where $\sigma$ again is implicitly restricted to the set of density operators of the system and  
\beq 
S':=\{Y: [X,Y]=0 ~ \forall X\in S\}
\eeq 
denotes the commutant of $S$, i.e., the set of operators that commute with all members of $S$. This generalised measure is, therefore, the distance from the state of the system to the set of states that commute with the operators in $S$.  Such measures are commonly used to quantify asymmetry resources, i.e., the degree to which the state is noninvariant with respect to members of $S$, when $S$ is either (i)~a unitary representation of some group~\cite{vacc1,vacc2,Gour,Marvian2014, HallPRX,ColesPRL, Hall2022}, or (ii)~an operator algebra~\cite{lara,Gao2020,Renner}.

Remarkably, while the cases of arbitrary sets, unitary group representations and operator algebras may appear to represent {\it prima facie} significant distinctions, this is in fact not so, provided that either $S$ or $S'$ in Eq.~(\ref{qds}) is closed under the adjoint operation (as is usually the case in practice~\cite{vacc1,vacc2,Gour,Marvian2014,lara,HallPRX,ColesPRL,Renner, Gao2020,Hall2022}).  For example, the standard asymmetry~(\ref{q1}) of a projective observable $X$ is equivalently represented in Eq.~(\ref{qds}) via any of the set $\{X\}$, the group of unitary operators $\{e^{ia X}\}$, and the algebra of operators that commute with $X$.
Links between these cases are therefore briefly clarified below for the interested reader, including their unification via the extension of quantumness to quantum channels, and applications to the standard asymmetry measure (other readers may wish to skip directly to Sec.~\ref{sec:uncert}).

\begin{proposition} \label{pro1}
	{\it
		If either $S$ or $S'$ is closed under the adjoint operation, then the quantumness measure $Q_d(S|\rho)$ in definition~(\ref{qds}) is invariant under the replacement of the set $S$ by the von Neumann algebra $S''$, or by the group $G_{S''}$ of unitary operators in $S''$ (where both replacements are closed under the adjoint operation).}
\end{proposition}
\begin{proof}
	Note that $S'$ is closed under the adjoint operation under either of the conditions of the Proposition, since  the closure of $S$ under the adjoint  implies that $[X^\dagger, Y]=0$ for all $X\in S$ and $Y\in S'$, which is equivalent to  $[X,Y^\dagger]=0$. Further, $S'$ is closed under multiplication and addition (since $[X,Y]=[X,Z]=0$ implies $[X,YZ]=[X,Y+Z]=0$), and contains the unit operator (since $[X,\id]\equiv0)$. Hence, $S'$ is an algebra,  and the closure of $S'$ under the adjoint operation guarantees that it is a von Neumann algebra~\cite{comm}. Now, any von Neumann algebra $\mathfrak{a}$ is equal to its double commutant $\mathfrak{a}''$, and its commutant $\mathfrak{a}'$ is also a von Neumann algebra~\cite{comm}. Hence, $S' = (S')''=(S'')'$, implying that $Q_d(S|\rho)=Q_d(S''|\rho)$ as claimed.  Moreover, any element of a von Neumann algebra  can be written as a linear combination of (at most four of) its unitaries~\cite{comm}, and hence an operator commutes with the members of the von Neumann algebra $S''$ if and only if it commutes with the members of the group $G_{S''}$ of unitary operators in $S''$. Thus, $S'=(S'')'=(G_{S''})'$, implying $Q_d(S|\rho)=Q_d(S''|\rho)=Q_d(G_{S''}|\rho)$ as claimed. Note each replacement is closed under the adjoint operation, since $S''$ is a von Neumann algebra and  $U^{-1}=U^\dagger$ for $U\in G_{S''}$. 
\end{proof}

Thus, for example, a distance-based measure of rotational asymmetry has the same value irrespective of whether $S$ in Eq.~(\ref{qds}) is chosen to be (a)~the set of rotation operators $\{J_x,J_y,J_z\}$, or (b)~the algebra generated by these operators, or (c)~the unitary representation $\{e^{i\bm J\cdot\bm n}\}$ of the rotation group~\cite{footrot}.  
This  general link between sets, groups and algebras of operators is further illuminated by extending the definition of $Q_d$ to quantum channels. 

In particular, for a given distance function $d(\rho,\sigma)$ and CPTP map $\phi$, define
\beq \label{qdphi}
Q_d(\phi|\rho):= \inf_\sigma d(\rho,\phi(\sigma)) .
\eeq
Thus, $Q_d(\phi|\rho)$ is the distance from $\rho$ to the set of output states of the channel, and may be interpreted as a measure of how closely $\rho$ can be prepared via the channel.
Further, this measure reduces to $Q_d(X|\rho)$ in Eq.~(\ref{qd}) for the channel $\phi_X(\sigma):=\sigma_X$.  It is therefore natural to ask, for the general case of an arbitrary set of operators $S$, whether there is a corresponding channel $\phi_S$ such that 
\beq \label{channel}
Q_d(S|\rho)=Q_d(\phi_S|\rho) 
\eeq
i.e., such that the range of $\phi_S$ is equal to the commutant $S'$ of $S$?
It turns out that, under the closure assumption of Proposition~\ref{pro1}, the answer is largely affirmative. 

First, if $S=\{X\}$ for some projective operator $X$, then 
\beq \label{phix}
\phi_S(\sigma)= \phi_X(\sigma):= \sum_x X_x\sigma X_x = \sigma_X 
\eeq
as noted above. Second, if the group $G_{S''}=\{U_g\}$ of unitary operators on $S''$ is compact, with normalised measure $dg$, then $\phi_S$ is the twirling map~\cite{vacc1,vacc2}
\beq \label{twirl}
\phi_S(\sigma)=\phi_{G_{S''}}(\sigma):=\int dg U_g\sigma U_g^\dagger .
\eeq
Third, if the system Hilbert space is finite dimensional, then 
\beq \label{expec}
\phi_S(\sigma) := \mathfrak{E}_{S'}(\sigma) 
\eeq
 where, for a given von Neumann algebra $\mathfrak{a}$, $\mathfrak{E}_\mathfrak{a}$ denotes the conditional expectation map  defined via~\cite{lara,comm}
\beq
\tr{XY}\equiv \tr{X\mathfrak{E}_{a}(Y)}
\eeq
for $X\in \mathfrak{a}$ and arbitrary $Y$. These maps agree on their common domains, and cover most situations of interest. 

Finally, these maps allow the explicit evaluation of the standard asymmetry $Q_1(S|\rho)$ in most cases.

\begin{proposition} \label{pro2}
	If $S=\{X\}$ for a projective observable $X$ and/or $G_{S''}$ is compact and/or the system Hilbert space is finite dimensional, for some set of operators $S$, and $d(\rho,\sigma)$ is the relative entropy distance function in Eq.~(\ref{d1}),  then the standard asymmetry measure $Q_1(S|\rho)$ defined in Eq.~(\ref{qds}) is given by
	\beq
	Q_1(S|\rho) = D_1(\rho\|\phi_S(\rho)) = H(\phi_S(\rho)) - H(\rho) .
	\eeq
\end{proposition}

\begin{proof}
	The result relies on the composition properties $\phi\circ\phi=\phi$, $\phi\circ f=f\circ\phi$ for any numerical function $f$, and the duality property  $\phi^*=\phi$, which may be checked to hold when $\phi$ is any of the maps in Eqs.~(\ref{phix})--(\ref{expec}).  In particular, these three properties, used in turn, yield 	$\tr{\rho\log\phi(\sigma)}=\tr{\rho\log\phi\circ\phi(\sigma)}
	=\tr{\rho\phi(\log\phi(\sigma))}
	=\tr{\phi(\rho)\log\phi(\sigma)}$. 
	Hence, definitions~(\ref{d1}) and~(\ref{qdphi}) give
	\begin{align}
		D_1(\rho\|\phi(\sigma))-D_1(\rho\|\phi(\rho))& = \tr{\phi(\rho)(\log\phi(\rho)-\log\phi(\sigma))}\nn\\
		&=D_1(\phi(\rho)\|\phi(\sigma)) \geq0, \nn
	\end{align}
	implying
	$Q_1(\phi|\rho) = D_1(\rho\|\phi(\rho))+\inf_\sigma D_1(\phi(\rho)\|\phi(\sigma))=D_1(\rho\|\phi(\rho))$. Thus, since Eq.~(\ref{channel}) holds under the stated conditions,
	\begin{align}
		Q_1(S|\rho)&=Q_1(\phi|\rho)=\tr{\rho\log\rho}-\tr{\rho\log\phi(\rho)}\nn\\
		&=-H(\rho) -\tr{\phi(\rho)\log\phi(\rho)}=H(\phi(\rho))-H(\rho)\nn
	\end{align}
	as required.
\end{proof}

Proposition~\ref{pro2} implies that $\phi_S(\rho)$ is the output state of $\phi_S$ that is closest to $\rho$, when `distance' is measured via relative entropy, and significantly generalises Eq.~(\ref{q1}) for the standard asymmetry of single observables. It has been previously given by Gao {\it et al.} for the case that $S$ is a von Neumann algebra on a finite Hilbert space~\cite{lara}. Note also that Corollary 2.3 of Gao {\it et al.} for strong subadditivity generalises to (using $(S\cup T)'=S'\cap T'$ and Proposition~\ref{pro1}) the uncertainty relation
\beq \label{gao}
Q_1(S|\rho) + Q_1(T|\rho) \geq Q_1(S\cup T|\rho), 
\eeq
for the standard asymmetry of any two sets $S$ and $T$ for which $\phi_S\circ\phi_T=\phi_T\circ\phi_S$. The following section is concerned with uncertainty relations for R\'enyi asymmetry.

\section{Tradeoff relations for  asymmetry and entropy}
\label{sec:uncert}

\subsection{R\'enyi asymmetry vs R\'enyi entropy}
\label{sec:tradeoff}

The R\'enyi asymmetry $Q_\alpha(X|\rho)$ in Eqs.~(\ref{renq}) and~(\ref{rengen}) represents the distance between the state of the system and a set of postmeasurement states for $X$, as measured via the sandwiched R\'enyi divergence in Eq.~(\ref{dalpha}). It is also a measure of the sensitivity of the state to transformations generated by $X$ (or its Naimark extension), and is the quantum component in the quantum-classical decomposition of R\'enyi entropy in Eq.~(\ref{renyidecomp}).  It reduces to the standard asymmetry given by the quantum relative entropy in Eq.~(\ref{q1}) for the case $\alpha=1$, and the general case has recently been found to have useful applications in the areas of quantum information, quantum metrology, quantum coherence, open quantum systems, and time-energy uncertainty relations~\cite{lara,ColesPRL,Hall2022,Chitambar,Gao2020}.

As the intrinsically quantum contribution to  R\'enyi entropy,  R\'enyi asymmetry should be expected to play a fundamental role in entropic uncertainty relations. This expectation is supported by uncertainty relation~(\ref{gao}) of Gao {\it et al.} for the standard asymmetry. In particular, for two  conjugate rank-1 projective observables $X$ and $Y$ on a $d$-dimensional Hilbert space, with $X_x=|x\rangle\langle x|$, $Y_y=|y\rangle \langle y|$ and $|\langle x|y\rangle|^2=d^{-1}$, one has $(\sigma_X)_Y=(\sigma_Y)_X=d^{-1}\id$ from Eq.~(\ref{sigmax}), implying $\phi_{\{X,Y\}}(\sigma)=d^{-1}\id$ from Eqs.~(\ref{qds}) and~(\ref{channel}), and Eq.~(\ref{gao}) 
simplifies to
\beq
Q_1(X|\rho) + Q_1(Y|\rho) \geq Q_1(\{X,Y\}|\rho) = \log d - H(\rho)
\eeq
with  $Q_1(\{X,Y\}|\rho)$ evaluated via Proposition~\ref{pro2}. 
Noting that $H(\rho_X)=H(X|\rho)$ and $H(\rho_Y)=H(Y|\rho)$ for such observables, and applying Eq.~(\ref{q1}), this is equivalent to the tradeoff relation,
\beq \label{conj}
Q_1(X|\rho) + H(Y|\rho) \geq \log d,
\eeq
between the standard asymmetry and Shannon entropy of the observables (and to the strong entropic uncertainty relation 
$H(X|\rho)+H(Y|\rho)\geq \log d + H(\rho)$
for the Shannon entropies of such observables~\cite{kor,Berta}).

A second example is the analogous tradeoff relation
\beq \label{number}
Q_\alpha(J_z|\rho) + H_\alpha(\Phi|\rho) \geq \log 2\pi,\qquad \alpha\geq\half,
\eeq
for angular momentum and rotation angle~\cite{Hall2022}. Noting the upper bound $Q_\alpha(J_z|\rho)\leq H_\beta(J_z|\rho)$ from Eq.~(\ref{renyidecomp}), this immediately implies and so is stronger than  the standard uncertainty relation $H_\alpha(J_z|\rho)+H_\beta(\Phi|\rho)\geq\log2\pi$ for the R\'enyi entropies of these observables~\cite{BB}.

It is shown here that similar tradeoff relations hold for arbitrary pairs of discrete observables, whether projective or nonprojective. However, the case of one projective and one arbitrary observable will be considered first, as the derivation is particularly simple for this case.

\begin{theorem} \label{thm1}
For a discrete projective observable $X$ with projection valued measure $\{X_x\}$, and an arbitrary discrete observable $Y$ with POVM $\{Y_y\}$, one has the tradeoff relation
\beq
Q_\alpha(X|\rho) + H_\alpha(Y|\rho) \geq -\log \max_{x,y}\lambda_{\max}(X_x Y_y X_x),~~ \alpha\geq\half,
\eeq
between R\'enyi asymmetry and entropy, where 
$\lambda_{\max}(A)$ is the maximum eigenvalue of Hermitian operator $A$. 
\end{theorem} 
\begin{proof} The proof is a generalisation of the derivation given by Korzekwa {\it et al.} of entropic uncertainty relation~(\ref{nondegen}) for rank-1  projective observables~\cite{kor}. First, define a `measure and discard' CPTP map $\varphi_Y$, from the Hilbert space of the system to the Hilbert space of a record system suitable for registering the result of a measurement of $Y$, via
\beq \label{phiy}
\varphi_Y(\rho):= \sum_y \tr{\rho Y_y}|y\rangle\langle y| ,
\eeq
where $\{|y\rangle\}$ is an orthonormal basis for the record system.
Equation~(\ref{sigmax}) and data processing inequality~(\ref{data}) for the sandwiched R\'enyi divergence then give
\begin{align}
D_\alpha(\rho\|\sigma_X) &\geq D_\alpha(\varphi_Y(\rho)\|\varphi_Y(\sigma_X))\nn\\
&=\frac{1}{\alpha-1}\log\sum_y \tr{\rho Y_y}^\alpha \tr{\sigma_X Y_y}^{1-\alpha}\nn\\
&=\frac{1}{\alpha-1}\log\sum_y \tr{\rho Y_y}^\alpha {\rm Tr} \big[\sigma \sum_x X_xY_yX_x\big]^{1-\alpha}\nn.
\end{align}
The final trace is over a sum of orthogonal subspaces, corresponding to the set of projections $X_x$. Hence, this trace is maximised, for a given value of $y$, when $\sigma$ is the eigenstate  corresponding to the maximum possible eigenvalue of $X_xY_yX_x$ over all $x$, and so is upper-bounded by $\mu=\max_{x,y}\lambda_{\max}(X_x Y_y X_x)$. Thus, noting that $\tfrac{1}{\alpha-1}\log\sum_ya(y) b(y)^{1-\alpha}$ is monotonic decreasing in $b(y)$ for $a(y), b(y), \alpha\geq 0$, the R\'enyi asymmetry in Eq.~(\ref{renq}) has the lower bound
\begin{align}
Q_\alpha(X|\rho) &= \inf_\sigma D_\alpha(\rho\|\sigma_X)\nn\\
&\geq \inf_\sigma \frac{1}{\alpha-1}\log\sum_y\tr{\rho Y_y}^\alpha \mu^{1-\alpha}\nn\\
&=-\frac{1}{1-\alpha}\log \sum_y p(y|\rho)^\alpha -\log\mu .
\label{proof1}
\end{align}
Recalling the definition of R\'enyi entropy in Eq.~(\ref{rendef}), this is equivalent to the statement of the theorem.
\end{proof}

Theorem~\ref{thm1} provides a strong link between asymmetry and entropy that significantly extends tradeoff relations~(\ref{conj}) and~(\ref{number}) for conjugate observables. It further leads to strengthened uncertainty relations for R\'enyi entropies, as discussed further below. Note that the proof of the Theorem is relatively elementary, relying on a simple manipulation of data-processing inequality~(\ref{data}) for the sandwiched R\'enyi divergence, suggesting it may also be used to obtain analogous tradeoff relations for other distance-based measures of quantumness or asymmetry. 

Noting that the asymmetry $Q_\alpha(X|\rho)$  vanishes for $[X,\rho]=0$ (consistent with Eq.~(\ref{q0}) for measures of quantumness),  Theorem~\ref{thm1} gives a lower bound for the R\'enyi entropy $H_\alpha(Y|\rho)$ in this case. A particular choice of $X$ yields the following interesting albeit weak corollary.

\begin{corollary} \label{cor1}
For an arbitrary discrete observable $X$ with POVM $\{X_x\}$, and state $\rho$ with  spectral decomposition $\rho=\sum_j e_j E_j$ (i.e., $E_j$ is the projection onto the eigenspace corresponding to eigenvalue $e_j$ of $\rho$), one has the lower bound
\beq
H_\alpha(X|\rho) \geq - \log \max_{x,j:e_j>0} \lambda_{\max}(E_jX_xE_j),~~~\alpha\geq\half,
\eeq
for R\'enyi entropy.
\end{corollary}
\begin{proof}
Define the projection onto the support of $\rho$ by $E:=\sum_{j:p_j>0} E_j$, implying  $\rho=E\rho E$, and let $Y_E$  be the projection of $Y$ onto the support of $\rho$, with POVM $\{EY_yE\}$. Replacing $X$ by $\rho$ and $Y$ by $Y_E$ in Theorem~\ref{thm1}, gives a vanishing asymmetry $Q_\alpha(\rho|\rho)=0$ via definition~(\ref{renq}) and a lower bound 
$\mu_{\rho Y_E} = \max_{y,j:p_j>0}\lambda_{\max}(E_jEY_yE E_j)=\max_{y,j:p_j>0}\lambda_{\max}(E_jY_y E_j) 
$
via $EE_j=E_j~(=0)$ for $p_j>0~(=0)$. Noting $H_\alpha(Y_E|\rho)=H_\alpha(Y|\rho)$ and changing notation from $Y$ to $X$ then yields the Corollary as desired.
\end{proof}

Corollary~\ref{cor1} bounds the R\'enyi entropy of any observable in terms of its degree of incompatibility with the system state. Note that it represents a slight sharpening of Eq.~(\ref{rhoxbound}) in the Introduction, where the latter does not limit the maximisation to $e_j>0$, and implies a state-independent lower bound, $-\log\lambda_{\max}(X_x)$, via $E_j\leq\id$. 
\blk However, the Corollary does not take the mixedness of the state into account,  making it relatively weak in comparison to bounds that do. 

For example, Corollary~\ref{cor1} is weaker than the known classical lower bound 
\beq \label{mono}
H_\alpha(X|\rho) \geq H_\infty(X|\rho) = -\log \max_x p(x|\rho)
\eeq
for R\'enyi entropy, following from the monontonic decreasing property $H_\alpha(X|\rho)\geq H_\beta(X|\rho)$ for $\alpha<\beta$~\cite{Harremoes,footmono},  where $H_\infty$ depends on mixedness via the eigenvalues $e_j$ of $\rho$. In particular, the operator $E_jX_xE_j$ is nonzero only on the $\tr{E_j}$-dimensional unit eigenspace of $E_j$ and so has at most $\tr{E_j}$ nonzero eigenvalues, implying that the sum of its eigenvalues, $\tr{E_jX_xE_j}$, is upper bounded by $\tr{E_j}\lambda_{\max}(E_jX_xE_j)$.  Hence, defining $p_j:=e_j \tr{E_j}$, then $\sum_{j:p_j>0} p_j=1$ and one has $p(x|\rho)=\sum_{j:e_j>0} e_j\tr{E_jX_xE_j}\leq \sum_{j:p_j>0} p_j\lambda_{\max}(E_jX_xE_j) \leq \max_{j:e_j>0} \lambda_{\max}(E_jX_xE_j)$. Maximising over $x$ immediately implies the bound in  Corollary~\ref{cor1} is never greater than the bound $H_\infty(X|\rho)$ in Eq.~(\ref{mono}), as claimed (with equality for pure states). 

A second example will be given in Corollary~\ref{cor4} below, which provides a strong lower bound for R\'enyi entropy that depends on mixedness rather than incompatibility, and which is stronger than both Corollary~\ref{cor1} and the classical bound~(\ref{mono}) for sufficiently mixed states.
\blk

The main technical result of this paper is a generalisation of Theorem~\ref{thm1} to all discrete observables, whether projective or nonprojective, via the generalised definition of asymmetry in Eq.~(\ref{rengen}).
\begin{theorem} \label{thm2}
For arbitary discrete observables $X$ and $Y$, with corresponding POVMs $\{X_x\}$ and $\{Y_y\}$, one has the tradeoff relation
\beq
Q_\alpha(X|\rho) + H_\alpha(Y|\rho) \geq -\log \mu_{XY},\qquad \alpha\geq\half,
\eeq
for asymmetry and entropy, where 
\beq
\mu_{XY} := \max_{x,y} \lambda_{\max}(X_x^{1/2}Y_y X_x^{1/2}) = \mu_{YX}
	\eeq
	denotes the maximum eigenvalue of $X_x^{1/2}Y_y X_x^{1/2}$ over $x$ and $y$.
\end{theorem} 
\begin{proof}
It is first necessary to generalise $Y$ to an observable $\tilde Y$ on the extended Hilbert space $\tilde{\cal H}$ for which the Naimark extension~(\ref{naimark}) is defined for $X$. In particular,	$\tilde Y$ is defined to correspond to the POVM $\{Y_y\}\cup \{P^\perp\}$ on $\tilde{\cal H}$, with $P^\perp:=\tilde\id-P$, and $\varphi_Y$ in Eq.~(\ref{phiy}) is generalised to the `measure and discard' map
\beq
\varphi_{\tilde Y}(\tilde{\rho}):= \sum_y \tr{\tilde{\rho} Y_y}|y\rangle\langle y| + \tr{\tilde\rho P^\perp} |y^\perp\rangle\langle y^\perp|, 
\eeq
for general states $\tilde \rho$ on $\tilde{\cal H}$, where $|y^\perp\rangle$ is the record state corresponding to POVM element $P^\perp$. For system state $\rho$ on $\cal H$ and general state $\tilde\sigma$ on $\tilde{\cal H}$ one then has, using  data processing inequality~(\ref{data}) and noting that $p(y^\perp|\rho)=\tr{\rho P^\perp}=0$ via $\rho=P\rho P$,
\begin{align}
	D_\alpha(\rho\|\tilde{\sigma}_{\tilde{X}}) &\geq D_\alpha(\varphi_{\tilde{Y}}(\rho)\|\varphi_{\tilde{Y}}(\tilde{\sigma}_{\tilde{X}}))\nn\\
	&=\frac{1}{\alpha-1}\log\sum_y \tr{\rho Y_y}^\alpha \tr{\tilde{\sigma}_{\tilde{X}} Y_y}^{1-\alpha}\nn\\
	&=\frac{1}{\alpha-1}\log\sum_y \tr{\rho Y_y}^\alpha {\rm Tr} \big[\tilde{\sigma} \sum_x \tilde X_xY_y\tilde X_x\big]^{1-\alpha}\nn
\end{align}
 similarly to the proof of Theorem~\ref{thm1}. Also similarly, noting that the $\tilde X_x$  are orthogonal projections, the final trace is maximised for a given value of $y$ by choosing $\tilde\sigma$ to be the eigenstate corresponding to the maximum over $x$ of the maximum eigenvalue of $\tilde X_xY_y\tilde X_x$. Now, 
\begin{align} 
\lambda_{\max}(\tilde X_xY_y\tilde X_x)&=\lambda_{\max}[(\tilde X_xY_y^{1/2})(\tilde X_xY_y^{1/2})^\dagger] \nn\\
&= \lambda_{\max}[(\tilde X_xY_y^{1/2})^\dagger(\tilde X_xY_y^{1/2})] \nn\\
&= \lambda_{\max}(Y_y^{1/2}\tilde X_xY_y^{1/2})] \nn\\
&= \lambda_{\max}(Y_y^{1/2}P\tilde X_x P Y_y^{1/2})] \nn\\
&= \lambda_{\max}(Y_y^{1/2} X_xY_y^{1/2})] \nn\\
&= \lambda_{\max}( X_x^{1/2} Y_yX_x^{1/2}),
\label{chain}
\end{align}
with the second and last lines following from $\lambda_{max}(AA^\dagger)=\lambda_{\max}(A^\dagger A)$, and the fourth and fifth lines  from $Y_y=PY_yP$ and Eq.~(\ref{naimark}), respectively. Hence,
\begin{align}
Q_\alpha(X|\rho) &=\inf_\sigma D_\alpha(\rho\|\sigma_{\tilde X}) \nn\\
	&\geq \inf_{\tilde \sigma}  D_\alpha(\rho\|\tilde\sigma_{\tilde X}) \nn\\
&\geq -\frac{1}{1-\alpha}\log \sum_y p(y|\rho)^\alpha -\log\mu_{XY}
\end{align}
similarly to Eq.~(\ref{proof1}). This is equivalent to the statement of the Theorem, noting that $\mu_{XY}=\mu_{YX}$ via the last two lines of Eq.~(\ref{chain}).
\end{proof}

Theorem~\ref{thm2} reduces to Theorem \ref{thm1} when $X$ is a projective observable, and underpins the strong uncertainty relation for R\'enyi entropies given in Theorem~\ref{thm3} below. Note that $\mu_{XY}$ can also be written as
\beq
\mu_{XY}= \max_{x,y} \| X_x^{1/2}Y_y^{1/2}\|_\infty^2 =\max_{x,y} \| Y_y^{1/2}X_x^{1/2}\|_\infty^2
\eeq
where $\|A\|_\infty=[\lambda_{\max}(A^\dagger A)]^{1/2}$ denotes the operator norm of $A$, i.e., the largest singular value of $A$.

Theorem~\ref{thm2} immediately provides a simple lower bound for R\'enyi asymmetry, as per the following corollary.
\begin{corollary} \label{corinsert}
For an arbitrary discrete observable $X$ with POVM $\{X_x\}$ one has the lower bound
\beq
Q_\alpha(X|\rho) \geq  -\log \max_x \lambda_{\max}(X_x) ,\qquad \alpha\geq\half,
\eeq
for the R\'enyi asymmetry of $X$.
\end{corollary}
\begin{proof}
Choose $Y=1$ in Theorem~\ref{thm2} and observe that $H_\alpha(1|\rho)\equiv0$.
\end{proof}

The lower bound in Corollary~\ref{corinsert} is trivial for projective observables, but typically nonvanishing otherwise, corresponding to the `hidden' asymmetry discussed following Eq.~(\ref{rengen}). For example, the lower bound is  $\log \tfrac{3}{2}$ for the qubit trine observable $T$ discussed following Corollary~\ref{cor4}.

Finally, it is worth noting here, as a preview to the next Section, that the standard uncertainty relation for R\'enyi entropies in Eq.~(\ref{renknown}) of the Introduction has a simple direct derivation via Theorem~\ref{thm2} for the case of Shannon entropy, i.e., $\alpha=\beta=1$, as per the following Corollary.

\begin{corollary} \label{cor2}
For arbitrary discrete observables $X$ and $Y$, with corresponding POVMs $\{X_x\}$ and $\{Y_y\}$, one has the entropic uncertainty relation
\beq
H(X|\rho) + H(Y|\rho) \geq -\log \mu_{XY}
\eeq
for Shannon entropy,  where 
\beq
\mu_{XY} := \max_{x,y} \lambda_{\max}(X_x^{1/2}Y_y X_x^{1/2}) = \mu_{YX}
\eeq
denotes the maximum eigenvalue of $X_x^{1/2}Y_y X_x^{1/2}$ over $x$ and $y$.
\end{corollary}
\begin{proof}
The uncertainty relation follows immediately for any pure state $\rho=|\psi\rangle\langle\psi|$,  from Eq.~(\ref{asymmgen}) and Theorem~\ref{thm2}. It then immediately follows for general states via the concavity of Shannon entropy.
\end{proof}

Corollary~\ref{cor2} is given for rank-1 projective observables in~\cite{maas} and for arbitrary observables in~\cite{Para}, and generalised to uncertainty relation~(\ref{renknown}) for R\'enyi entropies in~\cite{Rastegin2010}. However, the proofs of Theorem~\ref{thm2} and Corollary~\ref{cor2} ultimately rely on data-processing inequality~(\ref{data}), rather than on Riesz's theorem~\cite{maas,Para,Rastegin2010}. This has the advantage of leading to the stronger uncertainty relation in Eq.~(\ref{uncertrenyi}) for R\'enyi entropies for very little additional work, as will be shown next.

\subsection{Strong uncertainty relations for R\'enyi entropies}
\label{sec:strong}

It is straightforward to strengthen Corollary~\ref{cor2}, so as to obtain uncertainty relation~(\ref{uncertrenyi}) of the Introduction, which takes the mixedness of the state into account. 

\begin{theorem} \label{thm3}
For arbitrary discrete observables $X$ and $Y$, with corresponding POVMs $\{X_x\}$ and $\{Y_y\}$, one has the entropic uncertainty relation
\begin{align}  \label{thm3a}
	H_\alpha(X|\rho)+H_\beta(Y|\rho) &\geq -\log \mu_{XY}\nn\\&~\,~~+\max\{C_\alpha(Y|\rho),C_\beta(X|\rho)\} 
\end{align}
for R\'enyi entropies with $1/\alpha+1/\beta=2$, where $C_\alpha(Z|\rho)$ is the classical contribution to $H_\beta(Z|\rho)$, defined via  quantum-classical decomposition~(\ref{renyidecomp}), and
\beq
\mu_{XY} := \max_{x,y} \lambda_{\max}(X_x^{1/2}Y_y X_x^{1/2}) = \mu_{YX}
\eeq
denotes the maximum eigenvalue of $X_x^{1/2}Y_y X_x^{1/2}$ over $x$ and $y$.
\end{theorem}
\begin{proof}
Substitution of the quantum-classical decomposition~(\ref{renyidecomp}) of $H_\beta(X|\rho)$ into Theorem~\ref{thm2} gives
\beq
H_\beta(X|\rho) +H_\alpha(Y|\rho) \geq -\log\mu_{XY} + C_\alpha(X|\rho)
\eeq
for $1/\alpha+1/\beta=2$. Comparing the inequality resulting from swapping $X$ with $Y$ in this expression with the inequality resulting from swapping $\alpha$ with $\beta$ then yields the Theorem as desired.
\end{proof}

Recalling that the classical component $C_\alpha(X|\rho)$ vanishes for pure states as per Eq.~(\ref{classgen}), it is seen that Theorem~\ref{thm3} improves on the standard uncertainty relation in Corollary~\ref{cor2} by taking the mixedness of the state into account.  For example, for the special case $\alpha=\beta=1$ one has the following corollary.
\begin{corollary} \label{cor3}
For an arbitrary rank-1 discrete observable $X$ and an arbitrary discrete observable $Y$, with respective POVM elements $\{|x\rangle\langle x|\}$ and $\{Y_y\}$, the Shannon entropies of $X$ and $Y$ satisfy the uncertainty relation
\beq
H(X|\rho) + H(Y|\rho) \geq -\log \max_{x,y} \langle x|Y_y|x\rangle +H(\rho).
\eeq
\end{corollary}
\begin{proof}
For projective observables this result is an immediate consequence of Eq.~(\ref{c1}) and Theorem~\ref{thm3} for $\alpha=\beta=1$. More generally, for any discrete observable $X$, with Naimark extension $\tilde X$, Eqs.~(\ref{c1}) and~(\ref{extension}) yield $C_1(X|\rho)=H(\rho)-\sum_xp(x|\rho)H(\tilde\rho_x)$, with $\tilde\rho_x:=\tilde X_x\rho\tilde X_x/p(x|\rho)$ defined on the extended Hilbert space.
Hence, $C_1( X|\rho)\leq H(\rho)$ for all observables $X$ (including $X=Y$), with equality when $\tilde X$ is rank-1.  But for a rank-1 observable $X$ one can always choose $\tilde X$ to be rank-1. In particular, expressing the projection $\tilde X_x$ as an orthogonal sum of rank-1 projections, $\tilde X_x=\sum_k |x,k\rangle\langle x,k|$, the Naimark extension property $X_x=P\tilde X_xP$ in Eq.~(\ref{naimark})  requires that $P|x,k\rangle=0$ for all but one value of $k$, $k=k_x$ say, implying one can replace $\tilde X$ by a rank-1 projective observable $\tilde X'$ on the extended Hilbert space $\tilde H'$ generated by the span of $\{|x,k_x\rangle\}$,  with $\tilde X'_x:=|x,k_x\rangle\langle x,k_x|$.
Hence, making such a choice, it follows that $C_1(X|\rho)=H(\rho)\geq C_1(Y|\rho)$.  Substitution into Theorem~\ref{thm3} with $\alpha=\beta=1$ then gives the Corollary as desired.
\end{proof}

Corollary~\ref{cor3} generalises  the known  uncertainty relation for Shannon entropies of rank-1  projective observables in Eq.~(\ref{nondegen}) of the Introduction~\cite{Berta, kor}, and is given in Eq.~(71) of~\cite{Coles2011} under the assumption of a finite system Hilbert space.  Note that $H(\rho)$ appears as the measure of mixedness in this case because it is an upper bound for classicality when $\alpha=\beta=1$ and is saturated for rank-1 observables.

Finally, Theorem~\ref{thm3} also yields a strong lower bound for the R\'enyi entropy of discrete observables that takes mixedness into account.

\begin{corollary} \label{cor4}
For an arbitrary discrete observable $X$ with POVM $\{X_x\}$ one has the lower bound
\beq
H_\alpha(X|\rho) \geq  -\log \max_x \lambda_{\max}(X_x) + C_{\frac{\alpha}{2\alpha-1}}(X|\rho) ,
\eeq
for R\'enyi entropy.
\end{corollary}

\begin{proof} Choose $Y=1$ in Theorem~\ref{thm3}, and note that  $0\leq C_\alpha(1|\rho)\leq H_\beta(1|\rho)=0$ via quantum-classical decomposition~(\ref{renyidecomp}) (which is valid for arbitrary observables as discussed in Sec.~\ref{sec:nonproj}). The Corollary also follows directly from the asymmetry lower bound in Corollary~\ref{corinsert},  again using Eq.~(\ref{renyidecomp}) (but with $\alpha$ and $\beta$ swapped).
\end{proof}

For the case of a rank-1 observable $X$ and $\alpha=1$ one has $C_1(X|\rho)=H(\rho)$ as per the proof of Corollary~\ref{cor3}, and  the lower bound in Corollary~\ref{cor4} reduces to 
\beq \label{reduced}
H(X|\rho) \geq  -\log \max_x \langle x|x\rangle + H(\rho)
\eeq
This also corresponds to choosing $Y=1$ in Corollary~\ref{cor3}, and is stronger than the alternative choice $Y=X$, which replaces $H(\rho)$ by $\half H(\rho)$. \blk For rank-1 projective observables (i.e, $\langle x|x\rangle\equiv1$), it reduces to the known property that the Shannon entropy of such observables is never less than the von Neumann entropy (since $p(x|\rho)=\sum_jS_{xj}p_j$ where $S_{xj}:= |\langle x|\psi_j\rangle|^2$ is a doubly stochastic matrix for any orthogonal decomposition $\rho= \sum_jp_j|\psi_j\rangle\langle\psi_j|$ of $\rho$).

 It is of interest to compare Corollary~\ref{cor4}, which gives a lower bound for R\'enyi entropy that depends on the sharpness of $X$ and the mixedness of $\rho$, with classical lower bound~(\ref{mono}).  As a simple example, consider the projective qubit observable $X=\bm\sigma\cdot \bm n$ for spin direction $\bm n$, and state $\rho=\half(\id+\bm \sigma \cdot\bm r)$ with Bloch vector $\bm r$. Corollary~\ref{cor4} then yields the lower bound
\beq \label{cor5ex}
H(X|\rho) \geq H(\rho) = -\frac{1+r}{2}\log \frac{1+r}{2} - \frac{1-r}{2}\log \frac{1-r}{2} ,
\eeq
 for the Shannon measurement entropy, with $r:=|\bm r|$.  In contrast, the corresponding classical lower bound in Eq.~(\ref{mono}) gives
 \beq \label{cor5ex2}
 H(X|\rho) \geq \log \frac{2}{1+ |\bm  n\cdot \bm r|} = \log \frac{2}{1+r|\cos\theta_{\bm n,\bm r}|},
 \eeq
 where $\theta_{\bm n,\bm r}$ is the angle between $\bm n$ and $\bm r$. Both lower bounds range from 0 for  $\bm r=\pm \bm n$ up to $\log 2$ for $\bm r=\bm 0$. However, as depicted in Fig.~\ref{fig1}, the bound in Eq.~(\ref{cor5ex}) is always stronger for sufficiently mixed states (as is easily proved noting that Eq.~(\ref{cor5ex}) is concave in $r$ and Eq.~(\ref{cor5ex2}) is convex in $r$, with equality at $r=0$). 

\begin{figure}[!t] 
	\includegraphics[width=0.45\textwidth]{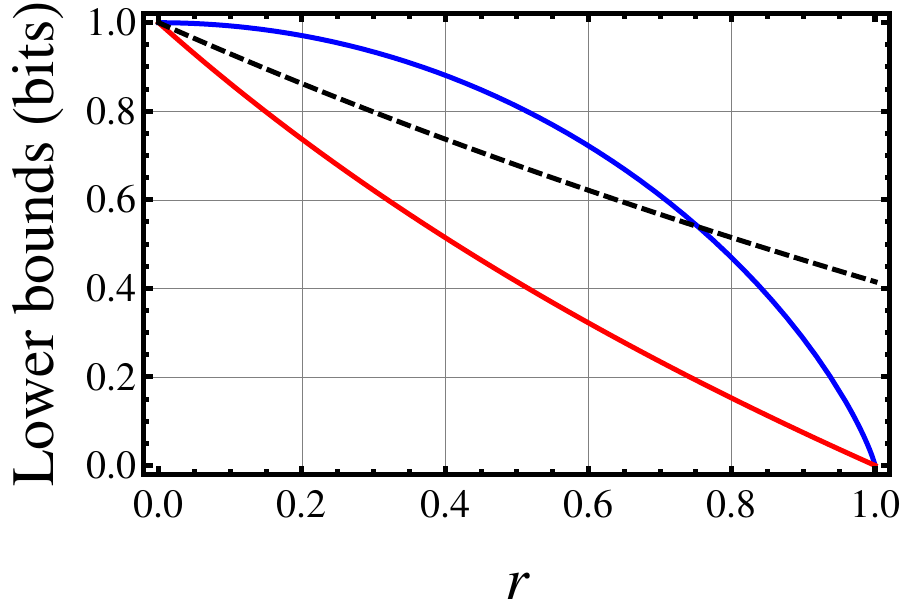}
	\caption{Lower bounds on measurement entropy for qubit states. The lower bound in Eq.~(\ref{cor5ex}) for qubit observable $X=\bm \sigma\cdot\bm n$, following from Corollary~\ref{cor4}, is plotted as a function of the Bloch vector length $r$ (upper solid blue curve), while the classical lower bound~(\ref{cor5ex2}) is plotted for the cases $\theta_{\bm n, \bm r}=0$, i.e., $\bm n\cdot \bm r=r$ (lower solid red curve), and 
	$\theta_{\bm n, \bm r}=\frac{\pi}{3}$, i.e., $\bm n\cdot\bm r=\half r$ (dashed black curve). Corresponding plots of the lower bounds for the measurement entropy of the trine observable $T$ in Eqs.~(\ref{trine2}) and~(\ref{trine}) may be obtained by adding $\log\frac{3}{2}$ to each lower bound. For both observables the classical lower bounds are seen to be weaker than the bound from Corollary~\ref{cor4} for sufficiently mixed states, as discussed more generally in the main text.}
	\label{fig1}
\end{figure}

 As an example for the case where $X$ is a nonprojective observable, consider the qubit trine observable $X=T$ with POVM elements $T_j=\tfrac{1}{3}(\id+\bm \sigma\cdot\bm m^{j})$, where $\bm m^{j}$ denotes the unit vector $(\cos(2\pi j/3),\sin(2\pi j/3),0)$ for $j=0,1,2$~\cite{peres,Holevo1973,foottrine}. Corollary~\ref{cor4} then gives the  lower bound
 \beq \label{trine2}
 H(T|\rho) \geq \log \frac{3}{2} + H(\rho) ,
 \eeq
 whereas the classical lower bound~(\ref{mono}) gives
 \beq \label{trine}
 H_\alpha(T|\rho) \geq \log \frac{3}{2} +\log \frac{2}{1+r\max_j \cos\theta_{\bm m^{j},\bm r}},
 \eeq 
where $\theta_{\bm n,\bm r}$ is the angle between $\bm m^{j}$ and $\bm r$. While both lower bounds range from a minimum value of $\log\tfrac{3}{2}$ for $\bm r=\bm m^{j}$, up to the maximum possible value of  $H_\alpha(T|\rho)=\log 3$ for  $\bm r$ orthogonal to the plane of trine directions, the bound in Eq.~(\ref{trine2}) is always stronger for sufficiently mixed states, as depicted in Fig.~\ref{fig1} (with the common term $\log\tfrac{3}{2}$ subtracted). For example, if $\bm r$ is  anti-aligned with one of the trine directions, then $\cos\theta_{\bm m^j,\bm r}$ is maximised for $\theta_{\bm m^j,\bm r}=\tfrac{\pi}{3}$ and Corollary~\ref{cor4} outperforms classical lower bound~\ref{mono} (and hence also Corollary~\ref{cor1}) for $r\lesssim 0.753$, corresponding to the intersection of the upper solid and dashed curves in Fig.~\ref{fig1}. 
\blk

More generally, for the case of an arbitrary observable $X$ and $\alpha=1$, Corollary~\ref{cor4} reduces to the lower bound
\begin{align}
H(X|\rho) &\geq -\log \max_x \lambda_{\max}(X_x)+ H(\rho)\nn\\
&\qquad -\sum_x p(x|\rho)H(\tilde\rho_x),
\end{align}
for  Shannon entropy, with $\tilde\rho_x:=\tilde X_x\rho\tilde X_x/p(x|\rho)$, for any  Naimark  extension $\tilde X$  of $X$ as per Eq.~(\ref{naimark}) (see also the proof of Corollary~\ref{cor3}). This case is intriguingly similar in form to the lower bound
\begin{align}
H(X|\rho_a) &\geq -\log \max_x \lambda_{\max}(X_x)+H(\rho_b)\nn\\
&\qquad -\sum_x p(x|\rho) H(\rho_{b|x})
\end{align}
for Shannon entropy in Eq.~(65) of~\cite{Coles2011} for any joint state $\rho_{ab}$, with $\rho_{b|x}:={\rm tr}_a[\rho_{ab}X_x]/p(x|\rho_a)$. It would be of interest to explore the connection between these bounds in future work. This is nontrivial given that $\rho_b=\sum_x p(x|\rho_a)\rho_{b|x}$ whereas $\rho\neq\sum_xp(x|\rho)\tilde\rho_x$. The key may lie in the related asymmetry lower bound given in Corollary~\ref{corinsert}.

\section{Discussion}
\label{sec:con}

The results weave two main threads together. The first is a general approach to quantum-classical decompositions of resource measures for observables (and sets of operators), based on the idea that pure states are the most resourceful (Sec.~\ref{sec:decompgen}). The quantum contribution to the resource measure arises from the noncommutativity or incompatibility of observables with the state, and the classical contribution from the mixedness of the state.  The approach unifies previous decompositions of variance and entropy, and generalises them to the decomposition of the symmetrised covariance matrix with respect to measures of quantum Fisher information, and the decomposition of Shannon entropy with respect to the conditional entropy of self-dual communication channels (Secs.~\ref{sec:skew}, \ref{sec:selfdual} and Appendix~\ref{appa:selfdual}), and is applied to R\'enyi entropy in particular (Sec.~\ref{sec:asymm}). In the latter case the quantum contribution is provided by the R\'enyi asymmetry, with its connection to R\'enyi entropy arising from duality relation~(\ref{asymmgen}) for pure states.

The second thread is the generalisation and unification of measures of quantumness and asymmetry to nonprojective observables (Sec.~\ref{sec:nonproj}), and to sets, groups and algebras of operators and to quantum channels (Sec.~\ref{sec:chann}). The R\'enyi asymmetry is a particular example of interest, as it generalises the standard asymmetry measure based on relative entropy and has applications in contexts as diverse as coherence, information, uncertainty relations, metrology, and open quantum systems.

The weaving of the two threads is in two natural stages. In the first stage, tradeoff relations between R\'enyi asymmetry and R\'enyi entropy are obtained, as per Theorems~\ref{thm1} and~\ref{thm2} (Sec.~\ref{sec:tradeoff}).  These have relatively elementary derivations, based on data processing inequality~(\ref{data}) for sandwiched R\'enyi divergences, and significantly generalise known relations~(\ref{conj}) and~(\ref{number}) for the case of conjugate observables. Corollaries~\ref{cor1}--\ref{cor2} are simple applications of these Theorems that yield  lower bounds for R\'enyi asymmetry and R\'enyi entropy, and recover previously known uncertainty relation~(\ref{renknown}) for the case of Shannon entropy.  All these results are independent of any consideration of quantum-classical decompositions.

It is in the second stage that the tradeoff relations for R\'enyi asymmetry are combined with the quantum-classical decomposition of R\'enyi entropy in Eq.~(\ref{renyidecomp}), to obtain the strong uncertainty relation for R\'enyi entropies in Theorem~\ref{thm3} (Sec.~\ref{sec:strong}).  This relation is valid for both projective and nonprojective observables, and improves on known uncertainty relation~(\ref{renknown}) by taking the mixedness of the system state into account via the classical contribution to R\'enyi entropy. Corollary~\ref{cor3} recovers a general result from Theorem~\ref{thm3} for the case of Shannon entropies, using the property that the classical contribution is bounded by the von Neumann entropy of the system for this case, while Corollary~\ref{cor4}  gives a strong bound for R\'enyi entropy that improves on
Corollary~\ref{cor1} and classical lower bound~(\ref{mono}) for sufficiently mixed states.
 
It is worthwhile remarking here on a point of terminology. The term `quantumness' has been used for the measure of noncommutativity $Q(X,Y,\dots|\rho)$ introduced in Sec.~\ref{sec:decompgen}, given its fundamental role as the inherently quantum contribution to uncertainty and other resource measures~\cite{Luo2005,LuoPRA2005,kor}. Note, however, that the nonvanishing of a quantum commutator has a direct classical analogue: the nonvanishing of a classical Poisson bracket. Hence, one can define classical analogues of asymmetry and the like (e.g., as a measure of the distance between a given phase space distribution and the set of phase space distributions invariant under canonical transformations generated by a given classical observable). Nevertheless, this does not extend to give a classical analogue of the approach to quantum-classical decompositions in Sec.~\ref{sec:decompgen}. In particular, this approach relies on the purification of a given quantum state on a larger Hilbert space, whereas there is no analogous purification of classical phase space distributions to a delta-distribution on a larger phase space.  Hence, the term `quantumness' may be justified as reflecting an essentially quantum feature in the context of quantum-classical decompositions.
 
There are various possible directions for future work, including the following. First, it is known that the bound $\mu_{XY}$ in known uncertainty relation~(\ref{renknown}) can be improved using majorisation properties of Schur-convex functions such as R\'enyi entropy~\cite{Fried2013,Ziggy2013}. This suggests looking for similar improvements to the lower bounds in Theorems~\ref{thm1}--\ref{thm3}. Perhaps, for example, the argument based on a monotonic decreasing function in the proof of Theorem~\ref{thm1} can be strengthened via a majorisation argument.

Second, it is of interest to seek analogous tradeoff relations and uncertainty relations for continuous observables such as position and momentum. Here the main technical issue is that the kets appearing in the spectral decompositions of such observables are nonnormalisable, so that the analogue of  $\sigma_X$ in Eq.~(\ref{sigmax}) is not a density operator. However,  it may be possible to take suitable limits of discrete observables to obtain results for this case, as has been done previously for Shannon entropic uncertainty relations~\cite{HallNJP,Hallholevo}.

Third,  whereas the measures of quantumness considered in the paper are focused on measures of uncertainty and asymmetry, consideration of other resources measures such as entanglement and Bell nonlocality, and related quantum-classical decompositions, is also a possible avenue of exploration. One might, for example, maximise some measure $Q(X_1\otimes X_2,Y_1\otimes Y_2,\dots|\rho)$ or $u(X_1\otimes X_2,Y_1\otimes Y_2,\dots|\psi)$ over sets of local observables to obtain a measure of bipartite entanglement for state $\rho$ or $\psi$, and apply the approach in Sec.~\ref{sec:decompgen} or the alternative approach in Appendix~\ref{appa:roof} to obtain corresponding quantum-classical decompositions. It would also be of interest to extend Proposition~\ref{pro2}  for the  standard asymmetry measure,  to sets $S$ of two or more (noncommuting) operators on an infinite Hilbert space, for the case that the group of unitary operators on the double commutant $S''$ is noncompact.

Fourth and finally, it should be possible to extend or relate the strong uncertainty relation in Theorem~\ref{thm3} to the presence of quantum memory or quantum side information, in which uncertainties are conditioned on systems with which the system of interest is correlated~\cite{Berta, Coles2011}. For example, the general method given by Coles {\it et al.} in~\cite{Coles2012} is applicable to R\'enyi conditional entropies of the form $H_\alpha(a|b):=-\inf_{\sigma_b}D_\alpha(\rho_{ab}\|\id_a\otimes\sigma_b)$, noting that the sandwiched R\'enyi divergence~(\ref{dalpha}) satisfies the required properties for $\alpha\geq\half$. Hence, noting further the duality property $H_\alpha(a|b)+H_\beta(a|c)=0$ for $1/\alpha+1/\beta=2$ and any pure state $\rho_{abc}$~\cite{Muller2013,Beigi2013}, Theorem~1 of~\cite{Coles2012} yields an uncertainty relation of the form 
\beq
H_\alpha(X|b) + H_\beta(Y|c) \geq -\log \mu_{XY},~~~\frac{1}{\alpha}+\frac{1}{\beta}=2,
\eeq
where $X$ and $Y$ are observables on the first component of a tripartite state $\rho_{abc}$,  $\rho_{Xb}:=\oplus_x {\rm Tr}_{ac}[\rho_{abc}X_x\otimes1_b\otimes1_c]\equiv\varphi_X(\rho_{ab})$, and $\rho_{Yc}:=\oplus_y {\rm Tr}_{ab}[\rho_{abc}Y_y\otimes1_b\otimes1_c]\equiv\varphi_Y(\rho_{ac})$ (see also Theorem~11 of~\cite{Muller2013}). It would be of interest to find a connection between quantum memory and the classical contribution to R\'enyi entropy via a comparison of this uncertainty relation with Theorem~\ref{thm3} (see also the related discussion at the end of Sec.~\ref{sec:strong}).

\acknowledgments
I am grateful to an anonymous referee for suggesting comparing Corollary~\ref{cor1} with classical lower bound~(\ref{mono}).
Having been recently diagnosed with pancreatic cancer, this may be one of my last physics papers. I therefore take this opportunity to thank the many colleagues I have interacted with over the years, both locally and nonlocally  :-) 

\appendix

\section{FURTHER EXAMPLES OF QUANTUM-CLASSICAL DECOMPOSITIONS}
\label{appa}

\subsection{Decomposition of Shannon entropy via a quantumness measure for self-dual channels}
\label{appa:selfdual}

A quantum-classical decomposition of Shannon entropy is obtained here that is quite different  to the decomposition of Korzekwa {\it et al.}~\cite{kor} discussed in Sec.~\ref{sec:selfdual}, based on a choice of quantumness measure that arises in the context of quantum information channels.

In particular, consider the communication channel defined by measuring observable $X$ with POVM $\{X_x\}$ at the receiver, on members of an ensemble of signal states with average density operator $\rho$,  for the case of the `dual' ensemble $\ens_{X,\rho}:=\{p_x;\rho_x\}$ defined by~\cite{Hugh93,Hall1997}
\beq \label{dual}
p_x:= \tr{\rho X_x},\qquad \rho_x:=\frac{\rho^{1/2}X_x\rho^{1/2}}{\tr{\rho X_x}} 
\eeq
(note that $\rho=\sum_x p_x\rho_x$, and the choice of $\rho_x$ is left arbitrary if $p_x=0$). 
This `self-dual' channel is invariant under source duality, i.e., interchange of the  the signal ensemble and the receiver measurement~\cite{Hugh93, Hall1997}. It provides a strong lower bound for the maximum information obtainable from a measurement of $X$ on an ensemble with density operator $\rho$~\cite{Hall1997}, and is closely related via the inverse map $p_x\rho_x\rightarrow X_x$ to the `pretty good measurement' used in the proof of the Holevo-Schumacher-Westmoreland theorem for the classical capacity of a quantum channel~\cite{nielsen}.  

Each measurement result $x$ of the self-dual channel gives some information about which signal state $\rho_{x'}$ was transmitted, with an average uncertainty quantified by the Shannon conditional entropy~\cite{nielsen},
\beq \label{condit}
H(X'|X) := H(X,X') - H(X) = H(X') - I(X:\ens_{X,\rho}) .
\eeq
Here $I(X:\ens_{X,\rho})$ is the mutual information of the channel, and $H(X')=-\sum_{x'} p_{x'}\log p_{x'}=H(X|\rho)$ is the entropy of the signal-state distribution $p_{x'}$ in Eq.~(\ref{dual}). Note that the joint distribution of $X$ and $X'$ for state $\rho$ follows via Eq.~(\ref{dual}) as
\beq \label{pxx}
p(x,x'|\rho):= \tr{\rho^{1/2}X_{x'}\rho^{1/2}X_{x}}.
\eeq

Now, if $X$ is a projective observable, then the conditional entropy of the self-dual channel provides a measure of quantumness satisfying assumption~(\ref{q0}) of Sec.~\ref{sec:decompgen} (since $[X,\rho]=0$ implies the signal states in Eq.~(\ref{dual})  are orthogonal, so that one has a noiseless classical channel with $H(X,X')=H(X)$). 
Further, assumptions~(\ref{assump1}) and~(\ref{assump2}) are also satisfied  (noting for a pure state $\rho=|\psi\rangle\langle\psi|$ that the signal states are identical and hence carry no information). Hence, noting $p(x,x')$ in Eq.~(\ref{pxx}) is invariant under Naimark extension~(\ref{naimark}) (since $P\rho P=\rho$), it follows that the conditional entropy of the self-dual channel in Eq.~(\ref{condit}) is a measure of quantumness for both projective and nonprojective observables, corresponding to 
\beq \label{qsd}
Q_{\rm sd}(X|\rho):= H(X|\rho) - I(X:\ens_{X,\rho}) .
\eeq
Recalling from~Eq.~(\ref{dual}) that the signal states are identical and carry no information when $\rho$ is pure, it follows $M_{\rm sd}(X|\rho)=H(\tilde X\otimes\id_a|\rho_\psi)=H(X|\rho)$ for any Naimark extension $\tilde X$ of $X$ and purification $\rho_\psi$ of $\rho$, and the approach of Sec.~\ref{sec:decompgen}  then yields  the   quantum-classical decomposition 
\beq \label{sd}
M_{sd}(X|\rho):= H(X|\rho) = Q_{\rm sd}(X|\rho) + C_{\rm sd}(X|\rho) 
\eeq
of Shannon entropy, with the classical component given by the mutual information of the self-dual channel, i.e.,
\beq 
C_{\rm sd}(X|\rho):=I(X:\ens_{X,\rho}) = 2H(X|\rho) -  H(X,X'|\rho) .
\eeq
Here the second equality follows by direct calculation~\cite{Hall1997}, where the last term is the Shannon entropy of the joint probability distribution in Eq.~(\ref{pxx}).

The decompositions of Shannon entropy in Eqs.~(\ref{qd}) and~(\ref{sd}) are seen to be rather different, although by construction they agree for projective observables when $[\rho,X]=0$ or $\rho=|\psi\rangle\langle\psi|$. It would be of interest to compare these decompositions in more detail elsewhere.

\subsection{An alternative roof-based approach to quantum-classical decompositions}
\label{appa:roof}

The method of constructing quantum-classical decompositions in Sec.~\ref{sec:decompgen} is guided by the level of accessibility, and hence usefulness, of a given quantum resource. In particular, the classical component arises due to lack of access to a purified state of the system. Here an alternative approach is briefly described, based on the direct decomposition of the system state into a mixture of pure states.  The approach is conceptually simple, but appears to be less general and technically more challenging than the approach given in the main text.

The starting point is a function $u(X,Y,\dots,|\psi)$ that quantifies a quantum resource for the pure states of a given system.  It is natural to consider such functions in various contexts, e.g., the entropy of the Schmidt coefficients for an entangled pure state.  The function $u$ is assumed to satisfies two minimal properties:
	\begin{itemize}
	\item[(a)] No `quantumness' when the system is in a simultaneous eigenstate of  observables $X,Y,\dots$, i.e., 
	\beq \label{prop1}
\qquad	u(X,Y,\dots|\psi)= 0~~{\rm if}~~ X\psi=x\psi,~Y\psi=y\psi, \dots ;
	\eeq
		\item[(b)] Invariance under unitary transformations, i.e.,
		\beq \label{prop2}
		\qquad u(UXU^\dagger,UYU^\dagger, \dots|U\psi) = u(X,Y,\dots|\psi)
		\eeq
		for any unitary transformation $U$.
\end{itemize}	
	
Now, let $\p_\rho$ denote the set of pure-state decompositions of density operator $\rho$, i.e.,
\beq
\p_\rho:=\bigg\{ \{p_j;\psi_j\}: p_j\geq0, \sum_j p_j|\psi_j\rangle\langle\psi_j| =\rho\bigg\},
\eeq
where $|\psi\rangle\langle\psi|$ denotes the density operator corresponding to $\psi$. The quantum-classical decomposition corresponding to function $u$ is then defined via
\begin{align} \label{def1}
	M(X,Y,\dots|\rho)&:= \sup_{\{p_j;\psi_j\}\in\p_\rho} \sum_j p_j  u(X,Y,\dots|\psi_j ) ,
	\end{align}
\begin{align}
	\label{def2}
	Q(X,Y,\dots|\rho)&:= \inf_{\{p_j;\psi_j\}\in\p_\rho} \sum_j p_j  u(X,Y,\dots|\psi_j) , 
	\end{align}
\begin{align}
	\label{def3}
		C(X,Y,\dots|\rho)&:=	M(X,Y,\dots|\rho)-	Q(X,Y,\dots|\rho).
\end{align}
Thus, $M$ and $Q$ are the concave and convex roofs of $u$, with the classical component given by their gap.

It is straightforward to check that the above construction satisfies the basic requirements discussed in Sec.~\ref{sec:decompgen}. Assumption~(\ref{q0}) follows from property~(\ref{prop1}), noting that commutativity implies there is a pure-state decomposition into simultaneous eigenstates, while assumptions~{\ref{assump1} and~(\ref{assump2}) are  direct consequences of property~(\ref{prop2}) and definitions~(\ref{def1})--(\ref{def3}), respectively. It is seen that the classical mixing contribution arises in this approach directly via the mixing of pure states, in contrast to the approach in Sec.~\ref{sec:decomp} in which the mixing is generated via a partial trace.

\subsubsection{Example: variance and minimal Fisher information}

As a first example, choose $u$ to be the variance of a projective observable, i.e.,
\beq
u_V(X|\psi):= {\rm Var}_\psi X = \langle X^2\rangle_\psi - \langle X\rangle_\psi^2.
\eeq
The concave and convex roofs of this function are given in~\cite{Toth} and~\cite{Yu}, respectively (see also~\cite{Toth2022}), and the above `roof' construction then yields the corresponding quantum-classical decomposition 
\beq
M_V(X|\rho) = {\rm Var}_\rho X = Q_V(X|\rho) + C_V(X|\rho),
\eeq
where
\beq
Q_V(X|\rho) = \half \sum_{j,k} \frac{|\langle j|[X,\rho]|k\rangle|^2}{\langle j|\rho|j\rangle +\langle k|\rho|k\rangle }.
\eeq
Here $\{|j\rangle\}$ is the set of eigenstates of $\rho$ and the summation runs over the eigenstates that generate nonzero denominators. The summation may be recognised, up to a constant factor, as the minimal quantum Fisher information, corresponding to the symmetric logarithmic derivative~\cite{Petz,Toth,Yu}, and hence this decomposition of the variance is a special case of the decomposition of the covariance matrix in Eq.~(\ref{covdecomp}) of Sec.~\ref{sec:decompex}. Thus, the 'roof' construction method picks out a preferred quantum Fisher information for the decomposition, in comparison to the method in the main text.

\subsubsection{Example: Shannon entropy and informational power}
	
More generally, the convex and concave roofs of a given function will  not  be amenable to analytic calculation, making this approach reliant on numerical methods. Consider, for example, the case where the function $u$ is the Shannon entropy of  observable $X$, i.e.,  $u_H(X|\psi):= H(X|\psi) = -\sum_x \langle\psi|X_x|\psi\rangle\log\langle\psi|X_x|\psi\rangle$. The roof construction then  gives the decomposition
\beq
M_H(X|\rho) = Q_H(X|\rho) + C_H(\rho) 
\eeq
with
\begin{align}
	M_H(X|\rho)&:= \sup_{\{p_j;\psi_j\}\in\p_\rho} \sum_j p_j H(X|\psi_j)\nn\\
	&= H(X|\rho) - \inf_{\ens\in\p_\rho} I(X:\ens) 
\end{align}
\begin{align}
	Q_H(X|\rho)&= \inf_{\{p_j;\psi_j\}\in\p_\rho} \sum_j p_j H(X|\psi_j)\nn\\
	&= H(X|\rho) - \sup_{\ens\in\p_\rho} I(X:\ens) ,
\end{align}
\begin{align}
	C_H(X|\rho)&:=\sup_{\ens\in\p_\rho} I(X:\ens)  - \inf_{\ens\in\p_\rho} I(X:\ens) ,
\end{align}
where $I(X:\ens)$ is the Shannon mutual information for observable $X$ and signal ensemble $\ens$. Hence the quantum component is directly related to the maximum mutual information of the observable over the set of pure-state signal ensembles for density operator $\rho$, also called the `informational power' of the observable for state $\rho$~\cite{infpower}, while the classical component is given by the gap between the maximum and minimum possible mutual informations.

 In general there are few analytic results for the maximum and minimum mutual informations of a pure state channel with fixed $X$ and $\rho$, and so the mutual information  must be maximised or minimised numerically. However, upper and lower bounds are provided by choosing particular pure state ensembles in $\p_\rho$, such as the Scrooge ensemble~\cite{Scrooge} or (for rank-1 observables) the self-dual ensemble in Eq.~(\ref{dual}) above, which can be used as estimates of the minimum and maximum values, respectively. The von Neumann entropy $H(\rho)$ also provides an estimate of the maximum value via the Holevo bound for mutual information~\cite{nielsen}.   Hence, for example, for a rank-1 observable $X$ one has the inequalities
 \beq
 H(X|\rho) - H(\rho) \leq Q_H(X|\rho) \leq Q_{\rm sd}(X|\rho),
 \eeq
 \beq
Q_{\rm sd}(X|\rho) - q(\rho) \leq   C_H(X|\rho) \leq H(\rho) ,
 \eeq
 where $Q_{\rm sd}(X|\rho)$ is the conditional entropy of the self-dual channel in Eqs.~(\ref{condit}) and~(\ref{qsd}), and $q(\rho)$ is the subentropy of $\rho$, i.e., the entropy of the Scrooge ensemble for $\rho$~\cite{Scrooge}.


\begin{thebibliography}{99}
\bibitem{footnai} The broader class of Naimark extensions with the projection $P$ replaced by an isometry~\cite{nai} is not considered here.
\bibitem{foottrine} Note that a Naimark extension $\tilde T$ for the trine observable $T$ discussed following Corollary~\ref{cor4} is obtained by adding an additional orthogonal basis element $|\tilde 0\rangle$ to the qubit Hilbert space and defining $|j\rangle:= (2/3)^{1/2}|\bm m^{j}\rangle+(1/3)^{1/2}|\tilde0\rangle$, $\tilde T_j:=| j\rangle \langle j|$ and $P=\tilde\id-|\tilde0\rangle\langle\tilde0|$,  where $|\bm m^{j}\rangle$ denotes the unit eigenstate of $\bm \sigma\cdot\bm m^{j}$ for $j=0,1,2$~\cite{peres}.
\bibitem{footrot} Another choice is $S=\{J_x,J_y\}$, noting for $X\in S'$ that $i\hbar[X,J_z]=[X,[J_x,J_y]]=0$ via the Jacobi identity.
\bibitem{footmono} The monotonic decreasing property of $H_\alpha$ with respect to $\alpha$ follows for a finite number of outcomes via Eqs.~(2) and Theorem~3 of~\cite{Harremoes}. More generally,  for $\alpha<\beta$ note that $f(x):=x^{(\alpha-1)/(\beta-1)}$ is convex (concave) for $\alpha<1$ ($\alpha>1$), yielding via definition~(\ref{rendef}) that $H_\alpha(X|\rho) = (1-\alpha)^{-1}\log \sum_x p(x|\rho) f(p(x|\rho)^{\beta-1})\geq (1-\alpha)^{-1}\log f(\sum_x p(x|\rho)^\beta) = H_\beta(X|\rho)$.
\blk
\end{thebibliography}
\end{document}